\documentclass[12pt]{article} 

\usepackage{color,amssymb,amsthm,amsmath,amsfonts,amsxtra,fullpage,graphicx,hyperref} 

\usepackage{fancyhdr}
\hypersetup{
    colorlinks, linkcolor={dark-blue},
    citecolor={blue}, urlcolor={blue}}

\usepackage{xcolor}
\definecolor{dark-red}{rgb}{0.4,0.15,0.15}
\definecolor{dark-blue}{rgb}{0,0,0.45}
\hypersetup{
    colorlinks, linkcolor={blue},
    citecolor={blue}, urlcolor={blue}}

\newcounter{constnum}
\def\mod{\text{ mod }} 

\def\constlabel#1{\newcounter{#1} \setcounter{#1}{\theconstnum}}

\def\const#1{%
\noindent\fbox{\parbox{\linewidth}{%
\refstepcounter{constnum}
{{\sc Construction \Roman{constnum}} \quad #1}}}}

\renewcommand{\qed}{\nobreak \ifvmode \relax \else
      \ifdim\lastskip<1.5em \hskip-\lastskip
      \hskip1.5em plus0em minus0.5em \fi \nobreak
      \vrule height0.75em width0.5em depth0.25em\fi}

\newtheorem{theorem}{Theorem}[section]
\newtheorem{lemma}[theorem]{Lemma}
\newtheorem{conjecture}[theorem]{Conjecture}

\theoremstyle{definition}
\newtheorem{example}[theorem]{Example}

\newtheorem{remark}{Remark}

\begin{document}

\centerline{{\LARGE Perfect Sequences and Arrays over the Unit Quaternions}}
\medskip
\centerline{\large Sam Blake}
\smallskip
\centerline{\large \it School of Mathematical Sciences, Monash University, Australia}
\bigskip
\bigskip
%
%
%

\def\ii{\textbf{\textit{i}}}
\def\jj{\textbf{\textit{j}}}
\def\kk{\textbf{\textit{k}}}

We introduce several new constructions for perfect periodic autocorrelation sequences and 
arrays over the unit quaternions. This paper 
uses both mathematical proofs and computer experiments to prove the (bounded) 
array constructions have perfect periodic autocorrelation. Furthermore, the first 
sequence construction generates odd-perfect sequences of unbounded lengths, 
with good ZCZ.\\

The quaternions were discovered by the Irish mathematician Sir William Rowan Hamilton 
in 1843. Hamilton was interested in the connection between complex numbers and 2-dimensional 
geometry. He tried in vain to extend the complex numbers to $\mathbb{R}^3$, only years later would it be 
discovered that there is no 3-dimensional normed division algebra. Hamilton's breakthrough
came when he extended the complex numbers to 4 dimensions. The quaternions are a 4-dimensional 
non-commutative division algebra. They can be thought of as an extension of the complex 
numbers, where $$\ii^2 = \jj^2 = \kk^2 = \ii\jj\kk = -1,$$ and they multiply as follows (row $\times$ column)
$$\begin{array}{c||c|c|c|c}
 \times  & 1 & \ii & \jj & \kk \\
\hline
\hline
 1 & 1 & \ii & \jj & \kk \\
\hline
 \ii & \ii & -1 & \kk & -\jj \\
\hline
 \jj & \jj & -\kk & -1 & \ii \\
\hline
 \kk & \kk & \jj & -\ii & -1
\end{array}.$$
Given the quaternions, $\textbf{q}_1 = a + b \ii + c \jj + d \kk$, where $a,b,c,d \in \mathbb{R}$, 
$\textbf{q}_1 = a_1 + b_1 \ii + c_1 \jj + d_1 \kk$ 
where $a_1,b_1,c_1,d_1 \in \mathbb{R}$, and  
$\textbf{q}_2 = a_2 + b_2 \ii + c_2 \jj + d_2 \kk$ where $a_2,b_2,c_2,d_2 \in \mathbb{R}$, 
we have 
\begin{align*}
\textbf{q}_1 + \textbf{q}_2 &= (a_1+a_2) + (b_1+b_2)\ii + (c_1+c_2)\jj + (d_1+d_2)\kk\\
\textbf{q}_1 \textbf{q}_2 &= (a_1 a_2-b_1 b_2-c_1 c_2-d_1 d_2) + 
(a_2 b_1+a_1 b_2-c_2 d_1+c_1 d_2) \ii + \\
& (a_2 c_1+a_1 c_2+b_2 d_1-b_1 d_2) \jj + (a_2 d_1+a_1 d_2-b_2 c_1+b_1 c_2) \kk\\
{\textbf{q}}^* &= a - b \ii - c \jj - d \kk \\
\|\textbf{q}\| &= \sqrt{\textbf{q}^*\textbf{q}} = \sqrt{a^2 + b^2 + c^2 + d^2} \\
\textbf{q}^{-1} &= \frac{\textbf{q}^*}{\|\textbf{q}\|^2}
\end{align*}
Over the quaternions, the notation 
$\frac{\textbf{q}_1}{\textbf{q}_2}$ is ambiguous, it could be left division, 
$\textbf{q}_2^{-1}\textbf{q}_1$, or right division, $\textbf{q}_1\textbf{q}_2^{-1}$. \\

If $\|\textbf{q}\| = 1$, then $\textbf{q}$ is a \textit{unit} quaternion. All the sequences 
and arrays defined in this chapter are restricted to \textit{simple unit} quaternions, which are given by 
$-1,1,-\ii,\ii,-\jj,\jj,-\kk,\kk$. \\

Moxey et al.\cite{Moxey2002} observed that as multiplication over the quaternions is 
non-commutative, the definition of autocorrelation becomes ambiguous. One could define a 
\textit{right} correlation: ${\theta^{\text{right}}}_{\textbf{s}}(\tau) = 
\sum\limits_{i=0}^{n-1}s_{i} s_{i+\tau}^*$, and a \textit{left} correlation: ${\theta^{\text{left}}}_{\textbf{s}}(\tau) = 
\sum\limits_{i=0}^{n-1}s_{i+\tau}^* s_{i}$. Thus we have two correlation measures over the
quaternions. These two correlations are not always equal. 

\begin{example}
We compute the left and right autocorrelations for the sequence, $\textbf{s} = [\ii,-\jj,-1,-1,\kk,1]$, and
show that the left and right autocorrelations of \textbf{s} are not equal for all shifts, $\tau$: 
\[
\begin{array}{c|r|r}
\tau & {\theta^{\text{left}}}_{\textbf{s}}(\tau) & {\theta^{\text{right}}}_{\textbf{s}}(\tau) \\
\hline
0 & 6 & 6\\
1 & \kk - \ii + \jj + 1 & 3\kk - \ii +\jj + 1\\
2 & \kk - \ii + 3\jj - 1 & \kk - \ii + \jj - 1\\
3 & -2 & -2\\
4 & -\kk + \ii -3\jj -1 & -\kk + \ii-\jj -1\\
5 & -\kk + \ii - \jj + 1 & -3\kk+\ii-\jj+1
\end{array}\]
\end{example}

A sequence is \textit{left perfect} if ${\theta^{\text{left}}}_{\textbf{s}}(\tau) = 0$ for all 
off-peak shifts. Similarly, for \textit{right perfect} sequences. \\

Kuznetsov\cite{Kuznetsov2010} showed that a sequence is left perfect if and only if the 
sequence is right perfect. Thus, a left or right
perfect sequence is simply referred to as a \textit{perfect sequence}. \\

Perfect quaternion sequences and arrays have a very short history. To the best of 
the authors knowledge, the first appearance in the literature of quaternion 
correlations was by Sangwine and Ell\cite{Sangwine1999} in 1999, where the Fourier
transform of colour images was computed via a quaternion Fourier transform. The autocorrelation was 
computed using the explicit (signal processing) definition, which is quadratic in complexity. \\ 

In 1999 Leukhin et al.\cite{Leukhin1999} described an application of quaternion sequences 
to optical signal processing involving photon echoes. There, a quaternion description of the 
polarisation state of the optical excitation is natural, and so is the quantum interaction within 
the material being analysed. Leukhin et al. discovered the first perfect sequence over the quaternions: 
\leftline{$\left[1+\jj,1+\jj,1+\jj,1+\jj,
        -\frac{1}{2} + \frac{\sqrt{3}\ii}{2} - \frac{\jj}{2} + \frac{\sqrt{3}\kk}{2},
        -\frac{1}{2} -\frac{\sqrt{3}\ii}{2} -\frac{\jj}{2} - \frac{\sqrt{3}\kk}{2},\right.$}
\rightline{$\left. 1+\jj, 
        -\frac{1}{2} - \frac{\sqrt{3} \ii}{2} - \frac{\jj}{2} - \frac{\sqrt{3} \kk}{2},
        -\frac{1}{2} + \frac{\sqrt{3} \ii}{2} - \frac{\jj}{2} + \frac{\sqrt{3} \kk}{2}\right]$} \\

In 2001 Pei et al.\cite{Pei2001} described an algorithm for the quaternion Fourier transform, 
convolution, and correlation by a 2D complex fast Fourier transform. \\

In 2010 Kuznetsov et al.\cite{Kuznetsov2010II} used the known method of taking the product of 
sequences of coprime lengths to construct a quaternion sequence of length
$5\times7\times9\times11\times13\times16\times17\times19\times23$. The sequences of smaller 
(odd) lengths were found by computer search, and of the form 
$[1, \jj, \textbf{s}, \textbf{s}_r, \jj, 1, \textbf{q}]$, where \textbf{s} is a simple quaternion
sequence and $q = \frac{a+b\ii+c\jj+d\kk}{2}$, where $a,b,c,d = \pm 1$. The even length sequence
is a Frank sequence\cite{Frank1962} of length 16 over 4 roots of unity. Based on
the length of the product sequence, Kuznetsov et al. conjectured
that these sequences exist for unbounded lengths. We believe this conjecture is questionable.
If the product sequences exist for unbounded lengths, then the individual sequences which form
the product sequence must exist for unbounded lengths. The longest sequence found by computer 
search is only of length 23 over an alphabet with 23 members, which far less than the square of the number of members in the sequence
alphabet, $529$ (see Mow's conjecture in Chapter 1). Stronger evidence for the conjecture could, 
for example, be the construction of sequences longer than $529$ which is not the product of 
smaller perfect sequences. \\

In 2012 Acevedo et al.\cite{Acevedo2012} discovered a link between Lee sequences\cite{Lee1992} and 
perfect quaternion sequences of lengths up to 98 which were discovered by the author in 2009\cite{SmithMiles2013}. 
As the Lee sequences exist for unbounded lengths, the corresponding sequences 
discovered by the author and subsequently generalised by Acevedo et al. are of unbounded length. While these 
sequences are over the quaternions, the single 
appearance of $\jj$ and no appearance of $-\jj$, $-\kk$, and $\kk$ suggests that for long 
sequences they are best thought of as \textit{near}-quaternion sequences, as the frequency of 
each member in the sequence 
alphabet asymptotically approaches a sequence over 4 roots of unity. \\

Following the paper of Acevedo et al.\cite{Acevedo2012}, Acevedo and Jolly\cite{Acevedo2013} 
extended the method of Arasu and de Launey\cite{Arasu2001} for constructing perfect arrays
of unbounded size over 4 roots of unity to quaternions. \\

In the following constructions we make frequent use of a well-known property of complex numbers.

\begin{lemma}\label{qnum1}
\hypertarget{quaternion1}
Let $m \in \mathbb{N}$ and $c \neq 0 \mod 4$, then $\sum\limits_{n=0}^{4m-1} \ii^{c n} = 0$.
\end{lemma}

\begin{proof}
As $\ii^{cn} = \ii^{cn \mod 4}$, the terms in the summation have a period of 4.  Thus
$$\sum\limits_{n=0}^{4m-1} \ii^{cn} = m \sum\limits_{n=0}^{3} \ii^{cn} = \left\{
\begin{array}{ccc}
m\left(\ii^0 + \ii^1 + \ii^2 + \ii^3\right) = m(1 + \ii - 1 - \ii) = 0 & \text{if} & c = 1 \mod 4\\
m\left(\ii^0 + \ii^2 + \ii^0 + \ii^2\right) = m(1 - 1 + 1 - 1) = 0 & \text{if} & c = 2 \mod 4\\
m\left(\ii^0 + \ii^3 + \ii^2 + \ii^1\right) = m(1 - \ii - 1 + \ii) = 0 & \text{if} & c = 3 \mod 4
\end{array}
\right.$$
\end{proof}

Of course, Lemma \ref{qnum1} holds if $\ii$ is replaced with $\jj$, or $\kk$. This
summation is the quaternion equivalent of the Gaussian sums for roots of unity for perfect sequences. That is, 
we reduce autocorrelation summations to terms containing this summation in order to show they are perfect. 

%
%
%
%

\section{Sequences of lengths 10, 14, 18, 26, 30, 38, 42, 50, 54, 62, 74, 82, 90, and 98}

The work presented in this section predates the work of Acevedo et al.\cite{Acevedo2012} (See
\cite{SmithMiles2013}). \\

In 2008 we used the Monash Campus Cluster\cite{MonashCG} to perform exhaustive searches for
perfect periodic autocorrelation sequences over $n$-tuples of 
$[-\ii,\ii,-\jj,\jj,-\kk,\kk]$. As cancellation occurs pairwise, perfect sequences over this 
alphabet only exist for even lengths. The exhaustive searches for small lengths produced
many sequences, for example

\[
\begin{array}{r|l}
\text{length} & \text{example sequence} \\
\hline
4 &  [-\kk, \ii, -\kk, -\ii]\\
6 & [\jj, \kk, -\jj, \kk, \jj, -\ii]\\
8 & [\kk, \ii, \kk, -\jj, -\kk, \ii, -\kk, -\jj]\\
10 & [-\ii, \kk, \ii, -\kk, \ii, -\jj, \ii, -\kk, \ii, \kk] \\
12 & -\\
14 & [-\ii, -\jj, \ii, \jj, \ii, \jj, \kk, \jj, \ii, \jj, \ii, -\jj, -\ii, \jj]\\
16 & [-\ii, -\ii, -\ii, -\jj, \ii, \kk, -\kk, -\jj, -\ii, \ii, \ii, -\jj, \ii, -\kk, \kk, -\jj]\\
\end{array}\]

After looking through hundreds of small length sequences a symmetry was observed; many of the perfect
sequences have the structure $$\textbf{q} = \left[-\ii, \textbf{s}, \kk, \textbf{s}_r \right],$$ where 
$\textbf{s} = [\alpha_0, \alpha_1, \cdots, \alpha_{n-1}]\times [\jj,\ii,\jj,\ii,\cdots, \jj, \ii]$, 
$\alpha_i = +1 \text{ or } -1$, $n$ is odd, and $\textbf{s}_r$ is the reverse of $\textbf{s}$. We refer
to \textbf{s} as the sub-sequence of \textbf{q}. Using this symmetry we can cut the search space for a 
sequence of length $n$ from $6^n$ to $2^{(n-2)/2}$, and thus we can find significantly longer sequences.
\[
\begin{array}{r|l}
14 & [-\ii,-\jj,-\ii,-\jj,\ii,\jj,-\ii,\kk,-\ii,\jj,\ii,-\jj,-\ii,-\jj]\\
18 & [-\ii, -\jj, -\ii, \jj, \ii, \jj, -\ii, \jj, \ii, \kk, \ii, \jj, -\ii, \jj, \ii, \jj, -\ii, -\jj]\\
26 & [-\ii, -\jj, -\ii, \jj, \ii, -\jj, \ii, \jj, \ii, \jj, -\ii, \jj, \ii, \kk, \ii, \jj, -\ii, \jj, \ii, \jj, \ii, -\jj, \ii, \jj, -\ii, -\jj]\\
30 & [-\ii, -\jj, -\ii, -\jj, \ii, \jj, \ii, \jj, -\ii, -\jj, \ii, \jj, \ii, -\jj, \ii, \kk, \ii, -\jj, \ii, \jj, \ii, -\jj, -\ii, \jj, \ii, \\
& \jj, \ii, -\jj, -\ii, -\jj]\\
38 & [-\ii, -\jj, -\ii, -\jj, -\ii, -\jj, \ii, \jj, \ii, -\jj, \ii, -\jj, -\ii, \jj, \ii, \jj, -\ii, -\jj, \ii, \kk, \ii, -\jj, -\ii, \jj, \ii, \jj, -\ii, \\
& -\jj, \ii, -\jj, \ii, \jj, \ii, -\jj, -\ii, -\jj, -\ii, -\jj]\\
42 & [-\ii, -\jj, -\ii, -\jj, -\ii, \jj, -\ii, \jj, \ii, \jj, -\ii, -\jj, \ii, -\jj, -\ii, \jj, \ii, -\jj, \ii, -\jj, -\ii, \kk, -\ii, -\jj, \ii, -\jj, \ii,\\ 
& \jj, -\ii, -\jj, \ii, -\jj, -\ii, \jj, \ii, \jj, -\ii, \jj, -\ii, -\jj, -\ii, -\jj]\\
50 & [-\ii, -\jj, -\ii, -\jj, -\ii, \jj, \ii, \jj, \ii, -\jj, \ii, -\jj, \ii, \jj, -\ii, -\jj, \ii, \jj, \ii, -\jj, \ii, \jj, -\ii, \jj, \ii, \kk, \ii, \jj,\\
& -\ii, \jj, \ii, -\jj, \ii, \jj, \ii, -\jj, -\ii, \jj, \ii, -\jj, \ii, -\jj, \ii, \jj, \ii, \jj, -\ii, -\jj, -\ii, -\jj]\\
54 & [-\ii, -\jj, -\ii, -\jj, -\ii, -\jj, \ii, \jj, -\ii, \jj, -\ii, -\jj, -\ii, -\jj, \ii, \jj, \ii, -\jj, \ii, \jj, -\ii, \jj, \ii, -\jj, -\ii, \jj, -\ii, \\
& \kk, -\ii, \jj, -\ii, -\jj, \ii, \jj, -\ii, \jj, \ii, -\jj, \ii, \jj, \ii, -\jj, -\ii, -\jj, -\ii, \jj, -\ii, \jj, \ii, -\jj, -\ii, -\jj, -\ii, -\jj]\\
62 & [-\ii, -\jj, -\ii, -\jj, \ii, \jj, \ii, -\jj, \ii, -\jj, -\ii, -\jj, \ii, -\jj, -\ii, \jj, -\ii, \jj, -\ii, -\jj, \ii, \jj, \ii, \jj, -\ii, -\jj, \\
& \ii, \jj, \ii, -\jj, \ii, \kk, \ii, -\jj, \ii, \jj, \ii, -\jj, -\ii, \jj, \ii, \jj, \ii, -\jj, -\ii, \jj, -\ii, \jj, -\ii, -\jj, \ii, -\jj, -\ii, -\jj,\\
&  \ii, -\jj, \ii, \jj, \ii, -\jj, -\ii, -\jj]\\
74 & [-\ii, -\jj, \ii, -\jj, \ii, -\jj, -\ii, -\jj, -\ii, -\jj, -\ii, -\jj, \ii, \jj, -\ii, \jj, \ii, \jj, -\ii, -\jj, \ii, \jj, \ii, -\jj, \ii, -\jj, \ii,\\
&  \jj, -\ii, \jj, -\ii, -\jj, \ii, -\jj, -\ii, \jj, \ii, \kk, \ii, \jj, -\ii, -\jj, \ii, -\jj, -\ii, \jj, -\ii, \jj, \ii, -\jj, \ii, -\jj, \ii, \jj, \ii, -\jj,\\
&  -\ii, \jj, \ii, \jj, -\ii, \jj, \ii, -\jj, -\ii, -\jj, -\ii, -\jj, -\ii, -\jj, \ii, -\jj, \ii, -\jj]\\
82 & [-\ii,-\jj,-\ii,\jj,\ii,-\jj,-\ii,-\jj,-\ii,-\jj,-\ii,-\jj,\ii,\jj,-\ii,\jj,\ii,-\jj,-\ii,-\jj,\ii,\jj,\ii,-\jj,\\
& -\ii,\jj,\ii, \jj,-\ii,\jj,-\ii,\jj,-\ii,\jj,\ii,\jj,\ii,-\jj,-\ii,-\jj,-\ii,\kk,-\ii,-\jj,-\ii,-\jj,\ii,\jj,\ii,\jj,\\
& -\ii,\jj,-\ii,\jj,-\ii,\jj, \ii,\jj,-\ii,-\jj,\ii,\jj,\ii,-\jj,-\ii,-\jj,\ii,\jj,-\ii,\jj,\ii,-\jj,-\ii,-\jj,-\ii,-\jj,\\
& -\ii,-\jj,\ii,\jj,-\ii,-\jj]\\
90 & [-\ii, -\jj, -\ii, -\jj, \ii, -\jj, \ii, -\jj, -\ii, -\jj, \ii, -\jj, -\ii, \jj, -\ii, -\jj, \ii, \jj, -\ii, \jj, \ii, \jj, -\ii, -\jj, \ii, \jj, \ii, \\
& -\jj, \ii, -\jj, -\ii, \jj, -\ii, \jj, -\ii, -\jj, -\ii, -\jj, \ii, -\jj, -\ii, \jj, -\ii, \jj, \ii, \kk, \ii, \jj, -\ii, \jj, -\ii, -\jj, \ii, -\jj, -\ii, \\
& -\jj, -\ii, \jj, -\ii, \jj, -\ii, -\jj, \ii, -\jj, \ii, \jj, \ii, -\jj, -\ii, \jj, \ii, \jj, -\ii, \jj, \ii, -\jj, -\ii, \jj, -\ii, -\jj, \ii, -\jj, -\ii, \\
& -\jj, \ii, -\jj, \ii, -\jj, -\ii, -\jj]\\
98 & [-\ii, -\jj, -\ii, -\jj, -\ii, -\jj, -\ii, \jj, -\ii, \jj, \ii, \jj, -\ii, -\jj, \ii, -\jj, -\ii, \jj, \ii, \jj, \ii, -\jj, \ii, -\jj, -\ii, -\jj, -\ii, \jj,\\
&  \ii, -\jj, -\ii, \jj, \ii, \jj, \ii, \jj, -\ii, -\jj, -\ii, \jj, -\ii, \jj, -\ii, -\jj, \ii, \jj, -\ii, \jj, \ii, \kk, \ii, \jj, -\ii, \jj, \ii, -\jj, -\ii, \jj,\\
&  -\ii, \jj, -\ii, -\jj, -\ii, \jj, \ii, \jj, \ii, \jj, -\ii, -\jj, \ii, \jj, -\ii, -\jj, -\ii, -\jj, \ii, -\jj, \ii, \jj, \ii, \jj, -\ii, -\jj, \ii, -\jj, -\ii, \\
& \jj, \ii, \jj, -\ii, \jj, -\ii, -\jj, -\ii, -\jj, -\ii, -\jj]\\
\end{array}
\]

The $\textbf{q}$ sequences have just a single \kk, so for long lengths the frequency of each member of the 
alphabet asymptotically approaches a sequence over 4 roots of unity. With detailed knowledge of this approach 
to constructing long perfect quaternion sequences, Acevedo and Hall extended this result by noticing a connection 
to Lee sequences\cite{Lee1992}. Acevedo and Hall chose not to reference this work\cite{Acevedo2012}.

%
%
%
%

\section{Quaternion sequences with the AOP}

Over roots of unity, a number of perfect sequence constructions exist which possess 
the AOP. These include the constructions of 
Heimiller-Frank\cite{Heimiller1961}\cite{Frank1962}, Milewski\cite{Milewski1983}, the generalised sequence construction 
of Mow\cite{Mow1996}, and a construction by the author\cite{Blake2013b}. 
The existence of perfect quaternion sequences with the AOP has not been previously considered. \\

We considered a construction of the form 
$$\textbf{S} = [S_{a,b}] = \ii^{\left\lfloor\frac{f(a,b)}{c}\right\rfloor}
\jj^{\left\lfloor\frac{g(a,b)}{d}\right\rfloor},$$ where $f(a,b)$ and $g(a,b)$ are bivariate 
polynomials with integer coefficients and $c$, $d$ are positive integers. The sequence is formed 
by enumerating row-by-row the array \textbf{S}. The size of the arrays considered were 
$2 \leq a, b \leq 32$, such that $a b > 16$. The integer coefficients of the polynomials and the 
denominators were less than 13. For each array size, $10\,000$ randomly selected polynomials and 
denominators were checked for the AOP. \\

The results of the search were interesting - one sequence of length 64 was found with the 
AOP. It was constructed from an array of size $8 \times 8$ with the index function 
$$\textbf{S} = [S_{a,b}] = \ii^{a\,b}\jj^{\left\lfloor\frac{a\,b}{2}\right\rfloor}.$$
The array and sequence is given by
$$
\left[
\begin{array}{cccccccc}
 1 & 1 & 1 & 1 & 1 & 1 & 1 & 1 \\
 1 & \ii & -\jj & -\kk  & -1 & -\ii & \jj & \kk  \\
 1 & -\jj & -1 & \jj & 1 & -\jj & -1 & \jj \\
 1 & -\kk  & \jj & \ii & -1 & \kk  & -\jj & -\ii \\
 1 & -1 & 1 & -1 & 1 & -1 & 1 & -1 \\
 1 & -\ii & -\jj & \kk  & -1 & \ii & \jj & -\kk  \\
 1 & \jj & -1 & -\jj & 1 & \jj & -1 & -\jj \\
 1 & \kk  & \jj & -\ii & -1 & -\kk  & -\jj & \ii \\
\end{array}
\right]
$$
and 
$$[1,1,1,1,1,1,1,1,1,\ii,-\jj,-\kk ,-1,-\ii,\jj,\kk ,1,-\jj,-1,\jj,1,-\jj,-1,$$
$$\jj,1,-\kk ,\jj,\ii,-1,\kk,-\jj,-\ii,1,-1,1,-1,1,-1,1,-1,1,-\ii,$$
$$-\jj,\kk ,-1,\ii,\jj,-\kk ,1,\jj,-1,-\jj,1,\jj,-1,-\jj,1,\kk,\jj,-\ii,-1,-\kk,-\jj,\ii].$$

The existence of this sequence is of interest, furthermore the fact that its length is the
square of the number of distinct elements in the sequence draws parallels to the 
Heimiller-Frank construction. Based on the search not finding longer 
sequences with the AOP we make the following conjecture. 

\begin{conjecture}\label{quaternion_AOP_bound_conjecture}
The longest perfect sequence over the simple unit quaternions with the AOP is of length 64.
\end{conjecture}

%
%
%
%

\section{Sequences of length $2^n$, for $0 < n < 7$}

\const{\constlabel{quatseq1}\hypertarget{quatseq1}
Let $0<n<6$, we construct a sequence, $\textbf{s}=[s_a]$, of length $2^n$, where 
$$s_a = \ii^{\left\lfloor\frac{a^2}{2^{n-1}}\right\rfloor}
        \jj^{\left\lfloor\frac{2a^2}{2^{n-1}}\right\rfloor},$$ for 
$0 \leq a < 2^n$.}

\bigskip

\begin{theorem}
Let \textbf{s} be the sequence from Construction
\Roman{quatseq1}. If $\tau$ is odd, then $\theta_{\textbf{s}}(\tau) = 0$.
\end{theorem}

\begin{proof}
The autocorrelation of \textbf{s} for shift $\tau$ is given by:
\begin{align}
\theta_{\textbf{S}}(\tau) &= \sum_{a=0}^{2^n-1} s_a s_{a+\tau}^* \nonumber\\
&= \sum_{a=0}^{2^n-1} \ii^{\left\lfloor\frac{a^2}{2^{n-1}}\right\rfloor}
\jj^{\left\lfloor\frac{2a^2}{2^{n-1}}\right\rfloor}
\left( 
\ii^{\left\lfloor\frac{(a+\tau)^2}{2^{n-1}}\right\rfloor}
\jj^{\left\lfloor\frac{2(a+\tau)^2}{2^{n-1}}\right\rfloor}
\right)^*\nonumber\\
\intertext{\indent Over the quaternions, we have, $(a b)^* = b^* a^*$:}
&=\sum_{a=0}^{2^n-1} \ii^{\left\lfloor\frac{a^2}{2^{n-1}}\right\rfloor}
\jj^{\left\lfloor\frac{2a^2}{2^{n-1}}\right\rfloor}
{\jj^{\left\lfloor\frac{2(a+\tau)^2}{2^{n-1}}\right\rfloor}}^*
{\ii^{\left\lfloor\frac{(a+\tau)^2}{2^{n-1}}\right\rfloor}}^*\nonumber\\
\intertext{\indent Over the quaternions, we have, $q^*=|q|^2 q^{-1}$:}
&=\sum_{a=0}^{2^n-1} \ii^{\left\lfloor\frac{a^2}{2^{n-1}}\right\rfloor}
\jj^{\left\lfloor\frac{2a^2}{2^{n-1}}\right\rfloor}
\jj^{-\left\lfloor\frac{2(a+\tau)^2}{2^{n-1}}\right\rfloor}
\ii^{-\left\lfloor\frac{(a+\tau)^2}{2^{n-1}}\right\rfloor}\nonumber\\
\intertext{\indent Change coordinates by letting $a=q2^{n-1}+r$:}
&=\frac{1}{2}\sum_{q=0}^3\sum_{r=0}^{2^{n-1}-1}
\ii^{\left\lfloor\frac{q^22^{2n-2}+2qr2^{n-1}+r^2}{2^{n-1}}\right\rfloor}
\jj^{\left\lfloor\frac{2q^22^{2n-2}+4qr2^{n-1}+2r^2}{2^{n-1}}\right\rfloor}
\times\nonumber\\
&\jj^{-\left\lfloor\frac{2q^22^{2n-2}+4qr2^{n-1}+2r^2
+4\tau q2^{n-1}+4\tau r + 2 \tau^2}{2^{n-1}}\right\rfloor}
\ii^{-\left\lfloor\frac{q^22^{2n-2}+2qr2^{n-1}+r^2
+2\tau q2^{n-1}+2\tau r + \tau^2}{2^{n-1}}\right\rfloor}\nonumber\\
&=\frac{1}{2}\sum_{q=0}^3\sum_{r=0}^{2^{n-1}-1}
\ii^{2^{n-1}q^2+2qr+\left\lfloor\frac{r^2}{2^{n-1}}\right\rfloor}
\jj^{2^{n}q^2 + 4qr + \left\lfloor\frac{2r^2}{2^{n-1}}\right\rfloor}
\times\nonumber\\
&\jj^{-2^{n}q^2 - 4qr - 4\tau q - \left\lfloor\frac{2r^2
+4\tau r + 2 \tau^2}{2^{n-1}}\right\rfloor}
\ii^{-2^{n-1}q^2 - 2qr - 2 \tau q - \left\lfloor\frac{r^2
+2\tau r + \tau^2}{2^{n-1}}\right\rfloor}\nonumber\\
&=\frac{1}{2}\sum_{q=0}^3\sum_{r=0}^{2^{n-1}-1}
\ii^{-2\tau q + \left\lfloor\frac{r^2}{2^{n-1}}\right\rfloor - 
\left\lfloor\frac{(r+\tau)^2}{2^{n-1}}\right\rfloor}
\jj^{-4\tau q + \left\lfloor\frac{2r^2}{2^{n-1}}\right\rfloor -
\left\lfloor\frac{2(r+\tau)^2}{2^{n-1}}\right\rfloor}\nonumber\\
&=\frac{1}{2}\left(\sum_{q=0}^3
\ii^{-2 \tau q}\right)
\left(\sum_{r=0}^{2^{n-1}-1}
\ii^{\left\lfloor\frac{r^2}{2^{n-1}}\right\rfloor - 
\left\lfloor\frac{(r+\tau)^2}{2^{n-1}}\right\rfloor}
\jj^{\left\lfloor\frac{2r^2}{2^{n-1}}\right\rfloor -
\left\lfloor\frac{2(r+\tau)^2}{2^{n-1}}\right\rfloor}\right)\label{quatseq1}
\end{align}
The leftmost summation in (\ref{quatseq1}) is zero for $\tau$ odd. 
\end{proof}

\begin{remark}
For $0<n<7$, we have confirmed by computer program that 
Construction \Roman{quatseq1} generates 
perfect sequences.
\end{remark}

\begin{example}
We construct the 6 perfect sequences generated by Construction \Roman{quatseq1}. The sequences of length 2, 4, and 8 are Zadoff-Chu and Milewski sequences.\\

\noindent length 2: \quad $[1, -\ii]$\\
\noindent length 4: \quad $[1, \jj, -1, \jj]$\\
\noindent length 8: \quad $[1,1,-\ii,-1,1,-1,-\ii,1]$\\
\noindent length 16:\quad $[1, 1, \jj, -\ii, -1, \ii, \jj, -1, 1, -1, \jj, \ii, -1, -\ii, \jj, 1]$\\
\noindent length 32:\quad $[1, 1, 1, \jj, -\ii, -\kk, -1, \ii, 1, -\ii, -1, \kk, -\ii, -\jj, 1, $\\
\hspace*{1in} $-1, 1, -1, 1, -\jj, -\ii, \kk, -1, -\ii, 1, \ii, -1, -\kk, -\ii, \jj, 1, 1]$\\
\noindent length 64:\quad $[1, 1, 1, 1, \jj, \jj, -\ii, -\kk, -1, -\jj, \ii, \kk, \jj, -\ii, -1, \ii, 1, -\ii, -1, \ii, \jj, -\kk, \ii, \jj, -1, \kk, -\ii, -\jj,$\\ 
\hspace*{1in} $\jj, -1, 1, -1, 1, -1, 1, -1, \jj, -\jj, -\ii, \kk, -1, \jj, \ii, -\kk, \jj, \ii, -1, -\ii, 1, \ii, -1, -\ii, \jj, \kk, \ii,$\\ 
\hspace*{1in} $-\jj, -1, -\kk, -\ii, \jj, \jj, 1, 1, 1]$
\end{example}

Extending Construction \Roman{quatseq1} to $n \geq 7$ generates good ZCZ sequences. 

\begin{example}
Let $n = 7$, then Construction \Roman{quatseq1} generates the following
sequence: 

$[1, 1, 1, 1, 1, 1, \jj, \jj, -\ii, -\ii, -\kk, -\kk, -1, -\jj, \ii, \kk, 1, \jj, -\ii, -\kk, -1, -\jj, \kk, 1, -\ii, -\kk, -\jj, \ii,$\\ 
\hspace*{0.3in} $1, -\ii, -1, \ii, 1, -\ii, -1, \ii, 1, -\ii, -\jj, \kk, -\ii, -1, \kk, \jj, -1, \kk, -\ii, -\jj, 1, -\kk, \ii, \jj, -1, \kk, -\kk, $\\
\hspace*{0.3in} $\ii, -\ii, -\jj, \jj, -1, 1, -1, 1, -1, 1, -1, 1, -1, 1, -1, \jj, -\jj, -\ii, \ii, -\kk, \kk, -1, \jj, \ii, -\kk, 1, -\jj, -\ii, \kk,$\\ 
\hspace*{0.3in} $-1, \jj, \kk, -1, -\ii, \kk, -\jj, -\ii, 1, \ii, -1, -\ii, 1, \ii, -1, -\ii, 1, \ii, -\jj, -\kk, -\ii, 1, \kk, -\jj, -1, -\kk, -\ii,$\\ 
\hspace*{0.3in} $\jj, 1, \kk, \ii, -\jj, -1, -\kk, -\kk, -\ii, -\ii, \jj, \jj, 1, 1, 1, 1, 1],$\\

which has the following (left and right) autocorrelations:\\

\leftline{128,0,0,0,0,0,0,0,16,0,0,0,0,0,0,0,0,0,0,0,0,0,0,0,16,0,0,0,0,0,0,0,0,0,0,0,0,0,0,0,-16,}
\centerline{0,0,0,0,0,0,0,0,0,0,0,0,0,0,0,-16,0,0,0,0,0,0,0,0,0,0,0,0,0,0,0,-16,0,0,0,0,0,0,0,0,0,0,}
\rightline{0,0,0,0,0,-16,0,0,0,0,0,0,0,0,0,0,0,0,0,0,0,16,0,0,0,0,0,0,0,0,0,0,0,0,0,0,0,16,0,0,0,0,0,0,0.}
\end{example}

%
%
%
%

\section{Arrays of size $2^n\times 2^n$, for $1<n<7$}

\const{\constlabel{quaternion2D1}\hypertarget{quaternion2D1}
Let $1<n<7$, we construct a 2-dimensional array, $\textbf{S} = 
\left[S_{a,b}\right]$ of size $2^n \times 2^n$ 
over the unit quaternions: $\{-1,1, -\ii, \ii, -\jj, \jj, -\kk, \kk\}$, 
where 
$$S_{a,b} = \ii^{\left\lfloor\frac{4 a b}{2^n}\right\rfloor}
\jj^{\left\lfloor\frac{4 a^2 b^2}{2^n}\right\rfloor}.$$}

\medskip

\begin{remark} 
For $1<n<7$, we have confirmed by computer program that 
Construction \Roman{quaternion2D1} generates 
perfect arrays.
\end{remark}

\begin{example}
Let $n=4$ as in Construction \Roman{quaternion2D1},
then we generate a $16\times 16$ perfect array: 
$$\left[
\begin{array}{cccccccccccccccc}
 1 & 1 & 1 & 1 & 1 & 1 & 1 & 1 & 1 & 1 & 1 & 1 & 1 & 1 & 1 & 1 \\
 1 & 1 & \jj & -1 & \ii & -\ii & \kk &
   \ii & -1 & -1 & -\jj & 1 & -\ii & \ii &
   -\kk & -\ii \\
 1 & \jj & \ii & \kk & -1 & -\jj & -\ii &
   -\kk & 1 & \jj & \ii & \kk & -1 & -\jj
   & -\ii & -\kk \\
 1 & -1 & \kk & -1 & -\ii & -\ii & \jj &
   -\ii & -1 & 1 & -\kk & 1 & \ii & \ii &
   -\jj & \ii \\
 1 & \ii & -1 & -\ii & 1 & \ii & -1 & -\ii & 1
   & \ii & -1 & -\ii & 1 & \ii & -1 & -\ii \\
 1 & -\ii & -\jj & -\ii & \ii & -1 & -\kk
   & -1 & -1 & \ii & \jj & \ii & -\ii & 1 &
   \kk & 1 \\
 1 & \kk & -\ii & \jj & -1 & -\kk & \ii &
   -\jj & 1 & \kk & -\ii & \jj & -1 &
   -\kk & \ii & -\jj \\
 1 & \ii & -\kk & -\ii & -\ii & -1 & -\jj
   & 1 & -1 & -\ii & \kk & \ii & \ii & 1 &
   \jj & -1 \\
 1 & -1 & 1 & -1 & 1 & -1 & 1 & -1 & 1 & -1 & 1 & -1 & 1 & -1 & 1 & -1
   \\
 1 & -1 & \jj & 1 & \ii & \ii & \kk &
   -\ii & -1 & 1 & -\jj & -1 & -\ii & -\ii &
   -\kk & \ii \\
 1 & -\jj & \ii & -\kk & -1 & \jj & -\ii
   & \kk & 1 & -\jj & \ii & -\kk & -1 &
   \jj & -\ii & \kk \\
 1 & 1 & \kk & 1 & -\ii & \ii & \jj & \ii
   & -1 & -1 & -\kk & -1 & \ii & -\ii & -\jj &
   -\ii \\
 1 & -\ii & -1 & \ii & 1 & -\ii & -1 & \ii & 1
   & -\ii & -1 & \ii & 1 & -\ii & -1 & \ii \\
 1 & \ii & -\jj & \ii & \ii & 1 & -\kk &
   1 & -1 & -\ii & \jj & -\ii & -\ii & -1 &
   \kk & -1 \\
 1 & -\kk & -\ii & -\jj & -1 & \kk & \ii
   & \jj & 1 & -\kk & -\ii & -\jj & -1 &
   \kk & \ii & \jj \\
 1 & -\ii & -\kk & \ii & -\ii & 1 & -\jj
   & -1 & -1 & \ii & \kk & -\ii & \ii & -1 &
   \jj & 1
\end{array}
\right]$$
\end{example}

%
%
%
%

\section{Arrays of size $2^{n+1}\times 2^{n+1} \times 2^{n+1} \times 2^{n+1}$, for $0<n<6$}

We now state a construction for $4$-dimensional arrays over the unit quaternions. \\

\const{\constlabel{quaternionC1}\hypertarget{quaternion1}
Let $0<n<6$, we construct a 4-dimensional array, $\textbf{S} = \left[S_{a,b,c,d} \right]$ of size $2^{n+1}\times 2^{n+1} \times 2^{n+1} \times 2^{n+1}$
over the unit quaternions: $\{-1,1, -\ii, \ii, -\jj, \jj, -\kk, \kk\}$, where 
$$S_{a,b,c,d} = \ii^{\left\lfloor\frac{a b}{2^{n-1}}\right\rfloor} \jj^{\left\lfloor\frac{b c}{2^{n-1}}\right\rfloor} \kk^{\left\lfloor\frac{c d}{2^{n-1}}\right\rfloor},$$ for $0 \leq a,b,c,d < 2^{n+1}$.}

\bigskip

\begin{theorem}\label{quaternionthm1}\hypertarget{quaternionthm1}
Let \textbf{S} be the array from Construction \Roman{quaternionC1}. If
$s_0 \neq 0 \mod 4$ or $s_3 \neq 0 \mod 4$, then
$\theta_{\textbf{S}}(s_0,s_1,s_2,s_3) = 0$.
\end{theorem}

\begin{proof}
The autocorrelation of \textbf{S} for shift $s_0, s_1, s_2, s_3$ is given by: 
\begin{align}
\theta_{\textbf{S}}&\left(s_0,s_1,s_2,s_3 \right)\nonumber\\ 
    &=\sum_{a=0}^{2^{n+1}-1}\sum_{b=0}^{2^{n+1}-1}\sum_{c=0}^{2^{n+1}-1}\sum_{d=0}^{2^{n+1}-1}
    S_{a,b,c,d}\,S_{a+s_0,b+s_1,c+s_2,d+s_3}^*\nonumber\\
&= \sum_{a=0}^{2^{n+1}-1}\sum_{b=0}^{2^{n+1}-1}\sum_{c=0}^{2^{n+1}-1}\sum_{d=0}^{2^{n+1}-1}
    \ii^{\left\lfloor\frac{a b}{2^{n-1}}\right\rfloor} \jj^{\left\lfloor\frac{b c}{2^{n-1}}\right\rfloor} 
    \kk^{\left\lfloor\frac{c d}{2^{n-1}}\right\rfloor} \,
    \left(\ii^{\left\lfloor\frac{(a+s_0) (b+s_1)}{2^{n-1}}\right\rfloor} 
    \jj^{\left\lfloor\frac{(b+s_1) (c+s_2)}{2^{n-1}}\right\rfloor} 
    \kk^{\left\lfloor\frac{(c+s_2) (d+s_3)}{2^{n-1}}\right\rfloor}\right)^*\nonumber\\
\intertext{\indent Over the quaternions, we have, $(a b)^* = b^* a^*$:}
&= \sum_{a=0}^{2^{n+1}-1}\sum_{b=0}^{2^{n+1}-1}\sum_{c=0}^{2^{n+1}-1}\sum_{d=0}^{2^{n+1}-1}
    \ii^{\left\lfloor\frac{a b}{2^{n-1}}\right\rfloor} 
    \jj^{\left\lfloor\frac{b c}{2^{n-1}}\right\rfloor}
    \kk^{\left\lfloor\frac{c d}{2^{n-1}}\right\rfloor} \,
    {\kk^{\left\lfloor\frac{(c+s_2) (d+s_3)}{2^{n-1}}\right\rfloor}}^*    
    {\jj^{\left\lfloor\frac{(b+s_1) (c+s_2)}{2^{n-1}}\right\rfloor}}^*
    {\ii^{\left\lfloor\frac{(a+s_0) (b+s_1)}{2^{n-1}}\right\rfloor}}^*\label{quatproof0}
\end{align}

Change coordinates by letting $a = q_a 2^{n-1} + r_a$, 
$b = q_b 2^{n-1} + r_b$, $c = q_c 2^{n-1} + r_c$, and $d = q_d 2^{n-1} + r_d$. Then
(\ref{quatproof0}) becomes: 
$$\theta_{\textbf{S}}\left(s_0,s_1,s_2,s_3 \right) = \sum_{q_a=0}^3\sum_{r_a=0}^{2^{n-1}-1}
\sum_{q_b=0}^3\sum_{r_b=0}^{2^{n-1}-1}
\sum_{q_c=0}^3\sum_{r_c=0}^{2^{n-1}-1}
\sum_{q_d=0}^3\sum_{r_d=0}^{2^{n-1}-1}
\ii^{\left\lfloor\frac{(q_a2^{n-1}+r_a)(q_b2^{n-1}+r_b)}{2^{n-1}}\right\rfloor}
\times$$
$$\jj^{\left\lfloor\frac{(q_b2^{n-1}+r_b)(q_c2^{n-1}+r_c)}{2^{n-1}}\right\rfloor}
\kk^{\left\lfloor\frac{(q_c2^{n-1}+r_c)(q_d2^{n-1}+r_d)}{2^{n-1}}\right\rfloor}
{\kk^{\left\lfloor\frac{(q_c2^{n-1}+r_c+s_2)(q_d2^{n-1}+r_d+s_3)}{2^{n-1}}\right\rfloor}}^*
\times$$
$${\jj^{\left\lfloor\frac{(q_b2^{n-1}+r_b+s_1)(q_c2^{n-1}+r_c+s_2)}{2^{n-1}}\right\rfloor}}^*
{\ii^{\left\lfloor\frac{(q_a2^{n-1}+r_a+s_0)(q_b2^{n-1}+r_b+s_1)}{2^{n-1}}\right\rfloor}}^*
$$
Expanding out the products, simplifying the {\it floor} functions, and
factoring sums of products to products of sums yields: 
$$\theta_{\textbf{S}}\left(s_0,s_1,s_2,s_3 \right) = 
\left(\sum_{q_a=0}^3 \ii^{-s_1q_a}\right)
\left(\sum_{q_b=0}^3 \ii^{-s_0q_b} \jj^{-s_2q_b}\right)
\left(\sum_{q_c=0}^3 \jj^{-s_1q_c} \kk^{-s_3q_c} \right)
\left(\sum_{q_d=0}^3 \kk^{-s_2q_d}\right)\times$$
$$\left(
\sum_{r_a=0}^{2^{n-1}-1}
\sum_{r_b=0}^{2^{n-1}-1}
\sum_{r_c=0}^{2^{n-1}-1}
\sum_{r_d=0}^{2^{n-1}-1}
\ii^{\left\lfloor\frac{r_ar_b}{2^{n-1}}\right\rfloor - 
\left\lfloor\frac{r_ar_b+s_0r_b+s_1r_a+s_0s_1}{2^{n-1}}\right\rfloor}
\times\right.$$
\begin{equation}
\left.\jj^{\left\lfloor\frac{r_br_c}{2^{n-1}}\right\rfloor - 
\left\lfloor\frac{r_br_c+s_1r_c+s_2r_b+s_1s_2}{2^{n-1}}\right\rfloor}
\kk^{\left\lfloor\frac{r_cr_d}{2^{n-1}}\right\rfloor - 
\left\lfloor\frac{r_cr_d+s_2r_d+s_3r_c+s_2s_3}{2^{n-1}}\right\rfloor}
\right)\label{quatautolast}
\end{equation}
By Lemma \ref{qnum1}, the first summation in (\ref{quatautolast}) 
is zero for $s_1 \neq 0 \mod 4$, similarly the fourth summation 
in (\ref{quatautolast}) is zero for $s_2 \neq 0 \mod 4$. Otherwise, 
for $s_1,s_2 = 0 \mod 4$, the second and third summations in 
(\ref{quatautolast}) are zero for $s_0, s_3 \neq 0 \mod 4$. 
\end{proof}

\begin{remark}
For $0<n<6$, Construction \Roman{quaternionC1}
generates perfect arrays. 
\end{remark}

We have confirmed by a computer program that the fifth summation in 
(\ref{quatautolast}) is zero for $0<n<6$ and for all shifts, and 
non-zero for many shifts for $n=6$. 

\begin{example}
We generate the smallest perfect array from 
Construction \Roman{quaternionC1} of size 
$4 \times 4 \times 4 \times 4$: 
{\small $$
\left[
\begin{array}{cccc}
 \left[
\begin{array}{cccc}
 1 & 1 & 1 & 1 \\
 1 & \kk  & -1 & -\kk  \\
 1 & -1 & 1 & -1 \\
 1 & -\kk  & -1 & \kk 
\end{array}
\right] & \left[
\begin{array}{cccc}
 1 & 1 & 1 & 1 \\
 \jj & \ii& -\jj & -\ii\\
 -1 & 1 & -1 & 1 \\
 -\jj & \ii& \jj & -\ii
\end{array}
\right] & \left[
\begin{array}{cccc}
 1 & 1 & 1 & 1 \\
 -1 & -\kk  & 1 & \kk  \\
 1 & -1 & 1 & -1 \\
 -1 & \kk  & 1 & -\kk 
\end{array}
\right] & \left[
\begin{array}{cccc}
 1 & 1 & 1 & 1 \\
 -\jj & -\ii& \jj & \ii\\
 -1 & 1 & -1 & 1 \\
 \jj & -\ii& -\jj & \ii
\end{array}
\right] \\
 \left[
\begin{array}{cccc}
 1 & 1 & 1 & 1 \\
 1 & \kk  & -1 & -\kk  \\
 1 & -1 & 1 & -1 \\
 1 & -\kk  & -1 & \kk 
\end{array}
\right] & \left[
\begin{array}{cccc}
 \ii& \ii& \ii& \ii\\
 \kk  & -1 & -\kk  & 1 \\
 -\ii& \ii& -\ii& \ii\\
 -\kk  & -1 & \kk  & 1
\end{array}
\right] & \left[
\begin{array}{cccc}
 -1 & -1 & -1 & -1 \\
 1 & \kk  & -1 & -\kk  \\
 -1 & 1 & -1 & 1 \\
 1 & -\kk  & -1 & \kk 
\end{array}
\right] & \left[
\begin{array}{cccc}
 -\ii& -\ii& -\ii& -\ii\\
 \kk  & -1 & -\kk  & 1 \\
 \ii& -\ii& \ii& -\ii\\
 -\kk  & -1 & \kk  & 1
\end{array}
\right] \\
 \left[
\begin{array}{cccc}
 1 & 1 & 1 & 1 \\
 1 & \kk  & -1 & -\kk  \\
 1 & -1 & 1 & -1 \\
 1 & -\kk  & -1 & \kk 
\end{array}
\right] & \left[
\begin{array}{cccc}
 -1 & -1 & -1 & -1 \\
 -\jj & -\ii& \jj & \ii\\
 1 & -1 & 1 & -1 \\
 \jj & -\ii& -\jj & \ii
\end{array}
\right] & \left[
\begin{array}{cccc}
 1 & 1 & 1 & 1 \\
 -1 & -\kk  & 1 & \kk  \\
 1 & -1 & 1 & -1 \\
 -1 & \kk  & 1 & -\kk 
\end{array}
\right] & \left[
\begin{array}{cccc}
 -1 & -1 & -1 & -1 \\
 \jj & \ii& -\jj & -\ii\\
 1 & -1 & 1 & -1 \\
 -\jj & \ii& \jj & -\ii
\end{array}
\right] \\
 \left[
\begin{array}{cccc}
 1 & 1 & 1 & 1 \\
 1 & \kk  & -1 & -\kk  \\
 1 & -1 & 1 & -1 \\
 1 & -\kk  & -1 & \kk 
\end{array}
\right] & \left[
\begin{array}{cccc}
 -\ii& -\ii& -\ii& -\ii\\
 -\kk  & 1 & \kk  & -1 \\
 \ii& -\ii& \ii& -\ii\\
 \kk  & 1 & -\kk  & -1
\end{array}
\right] & \left[
\begin{array}{cccc}
 -1 & -1 & -1 & -1 \\
 1 & \kk  & -1 & -\kk  \\
 -1 & 1 & -1 & 1 \\
 1 & -\kk  & -1 & \kk 
\end{array}
\right] & \left[
\begin{array}{cccc}
 \ii& \ii& \ii& \ii\\
 -\kk  & 1 & \kk  & -1 \\
 -\ii& \ii& -\ii& \ii\\
 \kk  & 1 & -\kk  & -1
\end{array}
\right]
\end{array}
\right]$$}
\end{example}

%
%
%
%

\section{Arrays of size $2^{n}\times 2^{n} \times 2^{n+1} \times 2^{n+1}$, for $0 < n < 6$}

We now state a second construction for $4$-dimensional arrays over the 
unit quaternions. The construction is very similar to Construction \Roman{quaternionC1}.

\bigskip

\const{\constlabel{quaternionC0}\hypertarget{quaternion0}
Let $n>0$ such that $n = 0 \mod 4$, we construct a 4-dimensional array, 
$\textbf{S} = \left[S_{a,b,c,d} \right]$ of size $2^{n}\times 2^{n} \times 2^{n+1} \times 2^{n+1}$
over the unit quaternions: $\{-1,1, -\ii, \ii, -\jj, \jj, -\kk, \kk\}$, where 
$$S_{a,b,c,d} = \ii^{\left\lfloor\frac{a b}{2^{n-1}}\right\rfloor} \jj^{\left\lfloor\frac{b c}{2^{n-1}}\right\rfloor} \kk^{\left\lfloor\frac{c d}{2^{n-1}}\right\rfloor},$$ for $0 \leq a,b,c,d < 2^{n+1}$.}

\bigskip

\begin{remark}
For $0<n<6$, we have confirmed by a computer program that 
Construction \Roman{quaternionC0} generates perfect arrays.
\end{remark}

\bigskip

%
%

In this paper we have found, by exhaustive computer search, perfect sequences over the simple unit 
quaternions of lengths 10, 14, 18, 26, 30, 38, 42, 50, 54, 62, 74, 82, 90, and 98. These sequences were subsequently 
generalised to unbounded lengths via the Lee sequences. We have shown that perfect quaternion sequences with the 
AOP exist, and subsequently perfect sequences and arrays exist over the simple unit quaternions. We have conjectured that the Frank
bound of the square of the number of elements in the alphabet of a perfect sequence extends to simple unit quaternions 
(Conjecture \ref{quaternion_AOP_bound_conjecture}). We have discovered constructions for perfect 2 and 4-dimensional 
arrays over the simple unit quaternions (Construction \Roman{quatseq1}, 
Construction \Roman{quaternionC1}, and Construction \Roman{quaternionC0}).

\bibliographystyle{abbrv}

\end{document}